\documentclass[runningheads,12pt]{llncs}
\usepackage[T1]{fontenc}
\usepackage{graphicx}

\usepackage[utf8]{inputenc}
\usepackage[T1]{fontenc}
\usepackage{mathpazo} 
\usepackage{pifont}
\usepackage[dvipsnames]{xcolor} 
\definecolor{ptblue}{RGB}{15,76,129} 
\definecolor{ptemerald}{HTML}{009473} 

\definecolor{cobalt}{rgb}{0.0, 0.28, 0.67}
\usepackage[margin=1.1in]{geometry}

\usepackage{soul}
\usepackage{url}
\usepackage[numbers]{natbib}

\makeatletter 
\renewcommand\@biblabel[1]{#1} 
\makeatother
\usepackage[hidelinks]{hyperref}
\hypersetup{
linktocpage,
colorlinks=true,
citecolor=cobalt, 
urlcolor=ptblue, 
linkcolor=Plum, 
}

\usepackage[small]{caption}
\usepackage{graphicx}
\usepackage{amsmath,amssymb}
\usepackage{booktabs}
\usepackage[ruled,linesnumbered,vlined]{algorithm2e}
\usepackage{multirow}
\SetKwRepeat{Do}{do}{while}

\usepackage{xcolor}

\DeclareMathOperator*{\argmax}{arg\,max}

\newtheorem{claimN}{Claim}

\usepackage{enumerate}
\newtheorem{observation}{Observation}

\usepackage{tikz}
  \usetikzlibrary{calc,shapes,decorations,matrix,arrows,positioning} 
  \usetikzlibrary{fit}
\usepackage{verbatim}

\makeatletter
\renewcommand{\@Opargbegintheorem}[4]{%
  #4\trivlist\item[\hskip\labelsep{#3#2\@thmcounterend}]}
\makeatother
\usepackage{amsmath,amssymb}
\let\doendproof\endproof
\renewcommand\endproof{~\hfill\qed\doendproof}

\begin{document}

\title{Neighborhood Stability in Assignments on Graphs}

\author{Haris Aziz\inst{1} \and
Grzegorz Lisowski\inst{2} \and
Mashbat Suzuki\inst{1} \and Jeremy Vollen\inst{1}}

 \authorrunning{Aziz, Lisowski, Suzuki, and Vollen}

\institute{UNSW Sydney
\and
AGH University}

\maketitle

\begin{abstract}
We study the problem of assigning agents to the vertices of a graph such that no pair of neighbors can benefit from swapping assignments -- a property we term \emph{neighborhood stability}. 
We further assume that agents' utilities are based solely on their preferences over the assignees of adjacent vertices and that those preferences are binary. 
Having shown that even this very restricted setting does not guarantee neighborhood stable assignments, we focus on special cases that provide such guarantees. 
We show that when the graph is a cycle or a path, a neighborhood stable assignment always exists for any preference profile.
Furthermore, we give a general condition under which neighborhood stable assignments always exist. 
For each of these results, we give a polynomial-time algorithm to compute a neighborhood stable assignment.

\end{abstract}


\section{Introduction}

An organization has drawn up a number of roles and decided which projects each role will work on \emph{a priori}.
Due to their passion for their organization's cause, the members of the organization are indifferent towards which projects they work on.
However, this indifference does not also apply to their feelings about their fellow team members: Each member is fond of some of their colleagues, but not all.
Furthermore, each member prefers roles in which they will work with more members they like.
How can the organization assign its members to the roles it has designed while avoiding the general chaos that may ensue if some pair of members prefer to swap roles?

The problem described above was first formalized by \citet{10.5555/3398761.3398979}.
Motivated by an analogous scenario in which a host must decide where to place each guest in a seating arrangement, they investigate the problem of assigning agents to a \emph{seat graph} in a manner that is \emph{stable} in the sense that no two agents would prefer to exchange assignments, i.e., no two agents form a \emph{blocking pair}.
The model has attracted follow-up work, resulting in a number of recent papers (see, e.g., \citet{DBLP:conf/ijcai/Ceylan0R23,berriaud2023stable,wilczynski2023ordinal}).
Besides the examples of assigning roles or seats, the problem can be summarized more broadly by its connection to hedonic games (see, e.g., \citet{AzSa15a,Cech08a}), in which agents are to be partitioned into disjoint coalitions according to their preferences over agents in their coalition. 
The literature on stability in hedonic games mostly focuses on complexity results, since stable outcomes are often not guaranteed to exist.
Indeed, \citet{berriaud2023stable1} recently investigated whether paths and/or cycles always admit a stable assignment in the seating arrangement model, and encountered numerous non-existence results, even under various preference restrictions.

Investigating the related topic of equilibria in Schelling games \citep{Schelling1969,Schelling1971}, \citet{bilo2022topological} found that restricting agents to local swaps may significantly enlarge the set of equilibria, implying equilibria existence for some graph classes which do not admit equilibria when allowing for arbitrary swaps.
In a similar vein, we focus in this work on a property we term \emph{neighborhood stability}, under which no two agents assigned to adjacent vertices form a blocking pair.
Neighborhood stability is quite natural in the seating arrangement and role assignment motivating settings.
Just as agents only derive utility from their neighbors in the seat graph because those are the individuals with which they can interact, it is intuitive to assume that swaps would only occur between neighboring agents due to the baseline level of interaction required to agree on a swap.
For instance, in a seating arrangement two agents that can communicate sufficiently to organize a swap must be seated close enough to enjoy each other's company (and thus must share an edge in the assignment).
In our role assignment motivating example, it may be that agents are only able to learn of others' roles by communicating with the other members they work with, and thus swaps can only happen between colleagues.   
In general, we note that neighborhood stability is especially well-motivated when the agents' information about the selected assignment is limited.

Searching for classes of instances for which stable assignments are guaranteed to exist, \citet{berriaud2023stable1} parameterize the problem by (1) the number of values an agent's utility for another agent can take on and (2) the number of classes of agents, where each agent belonging to the same class has identical preferences and is indifferent between any pair of agents belonging to the same class.

They show that even for the case of two-valued preferences and at least five classes, there are instances both for cycles and for paths that do not admit any stable assignments.\footnote{We note that the only positive stability result for an unrestricted number of agent classes due to \citet{berriaud2023stable1} comes as a result of weakening stability to only account for blocking pairs separated by a distance of at most two. Thus, neighborhood stability's approach to relaxing stability has precedent in our problem of interest.}

In this work, we allow for arbitrarily many classes of agents, and focus on existence of neighborhood stability rather than stability without restrictions on the location of the blocking pairs.
Repurposing an argument used by \citet{berriaud2023stable1}, it can be shown that there exist three-valued non-negative preferences for four agents such that no assignment is neighborhood stable on a cycle.\footnote{See the proof of Theorem 8 from \citet{berriaud2023stable1}. For $n=4$, every swap mentioned in the proof is between adjacent agents, and thus this instance fails neighborhood stability.}
Given this, we restrict our focus to \emph{binary preferences}.\footnote{It is without loss of generality to restrict two-valued non-negative preferences to binary preferences. And if we allow preferences that can take on negative values, it is easy to construct a counterexample of neighborhood stability on a path of length three. Specifically, let agents' utilities be either 1 or -1 and let the friendship graph be a cycle.}
Binary (sometimes referred to as \emph{dichotomous}) preferences can be interpreted naturally as characterizing each agent's preferences by a set of agents they `approve', and have been studied in various domains including committee voting \citep{LaSk23a}, matching (see, e.g., \citet{bogomolnaia2004random}), and item allocation (see, e.g., \citet{halpern2020fair}). 
Binary preferences allow us to represent agents' preferences simply as a directed graph, and as a result, our problem takes on the character of a seemingly fundamental graph theoretic problem.
As we will show, restricting to binary preferences still does not guarantee existence of neighborhood stability (see Figure~\ref{fig:counterexample} for an example with six agents).
Thus, in this paper, we investigate the following question:
\begin{quote}
	\centering
\emph{For which classes of seat graphs is a neighborhood stable assignment guaranteed to exist under binary preferences?}
\end{quote}

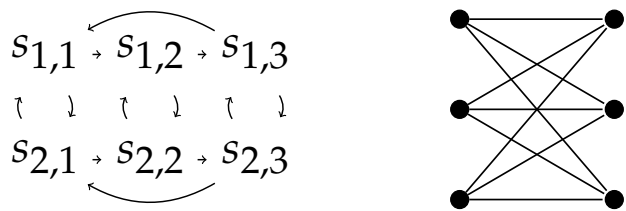
\begin{figure}[t]
\centering
\scalebox{0.7}{
\begin{tikzpicture}
         [->,shorten >=1pt,auto,node distance=2cm, semithick]
         \node[shape=circle] (A)  {\huge $s_{1,1}$};
         \node[shape=circle, below of=A] (B)  {\huge $s_{2,1}$};
         \node[shape=circle, right of =A] (C)  {\huge $s_{1,2}$};
         \node[shape=circle, below of =C] (D)  {\huge $s_{2,2}$};
         \node[shape=circle, right of =C] (E) {\huge $s_{1,3}$};
         \node[shape=circle, below of =E] (F)  {\huge $s_{2,3}$};
         
         \draw [thick,->, bend left] (A) to  (B);
         \draw [thick,->, bend left] (B) to  (A);
         \draw [thick,->, bend left] (C) to  (D);
         \draw [thick,->, bend left] (D) to  (C);
         \draw [thick,->, bend left] (E) to  (F);
         \draw [thick,->, bend left] (F) to  (E);
         \draw [thick,->] (A) to  (C);
         \draw [thick,->] (C) to  (E);
         \draw [thick,->, bend right] (E) to  (A);
         \draw [thick,->] (B) to  (D);
         \draw [thick,->] (D) to  (F);
         \draw [thick,->, bend left] (F) to  (B);
\end{tikzpicture}}
~
\hspace{5em}
\scalebox{0.8}{
\begin{tikzpicture}
            [->,shorten >=1pt,auto,node distance=1.5cm, semithick]
          \node[shape=circle,draw=black, fill=black] (1)  {};

          \node[shape=circle, draw=black, fill=black, below of=1] (2)  {};
          \node[shape=circle, draw=black, fill=black, below of=2] (3)  {};
          
          \node[shape=circle,draw=black, fill=black, right=2.25cm of 1] (4)  {};
          \node[shape=circle,draw=black, fill=black, below of =4] (5)  {};
          \node[shape=circle,draw=black, fill=black, below of =5] (6)  {};
          
          \draw [thick,-] (1) to  (4);
          \draw [thick,-] (1) to  (5);        
          \draw [thick,-] (1) to  (6);

          \draw [thick,-] (2) to  (4);
          \draw [thick,-] (2) to  (5);        
          \draw [thick,-] (2) to  (6);
          
          \draw [thick,-] (3) to  (4);
          \draw [thick,-] (3) to  (5);        
          \draw [thick,-] (3) to  (6);
\end{tikzpicture}}
     \caption{A preference (left) and seat graph (right) with $n=6$ which admits no neighborhood stable assignment. Theorem~\ref{thm:negativeEx} generalizes this example to obtain a counterexample of arbitrary size.} \label{fig:negativeEx}
     \label{fig:counterexample}
\end{figure}

\paragraph{Contributions}
When the seat graph is a cycle and agents have binary preferences, \citet{berriaud2023stable} gave an instance with five agents for which no assignment is stable, even if only blocking pairs at distance at most two from each other are considered. 
In our central result, we show that if we marginally 
weaken this requirement to neighborhood stability, stable assignments are guaranteed to exist on cycles of arbitrary length. 
Furthermore, we give a polynomial-time algorithm to compute such assignments.
Our algorithm introduces a novel technique using path partitions of the preference profile, and may be of independent interest. 
When the seat graph is a path and agents have binary preferences, \citet{berriaud2023stable} showed that swap dynamics are guaranteed to converge to an assignment in which no blocking pair is within distance at most two of each other. 
However, they gave evidence that there may be instances for which convergence requires exponential time, and left open whether an assignment satisfying this property can be computed in polynomial time.
We answer this question affirmatively, giving a polynomial-time algorithm for such an assignment on path seat graphs.

We complement these results with a general counterexample which shows that neighborhood stable assignments are not always guaranteed to exist, even for instances with as few as six agents.
Despite this negative result, we give a positive result for a general family of seat graphs.
We give a polynomial-time algorithm to compute a neighborhood stable assignment provided the size of a directed feedback vertex set of the preference graph is upper bounded by the number of leaf nodes in the seat graph.
To give one example, this means that when the preference graph is planar, our algorithm computes a neighborhood stable assignment for any seat graph with at least 60\% leaf nodes.

\subsection*{Related Work}

The paper most relevant to our work is from \citet{berriaud2023stable}. 
On a related note, \citet{bullinger2024topological} studied stability of assignments in which agents' utilities depend on the distance to other agents. They focus on jump stability, wherein no agent benefits from moving to an unassigned vertex. In this setting, they show that a jump stable assignment always exists under acyclic preferences, or when they are symmetric. We note that symmetric preferences also guarantee neighborhood stability, which can be shown analogously to the reasoning by \citet{bullinger2024topological}. 

We note that several studies that are closely related to our investigation focused on computational complexity aspects of assignments. For instance,  \citet{10.5555/3398761.3398979} investigated the problem of computing stable seat arrangements, focusing on the (parametrized) complexity angle. Furthermore, recently \citet{DBLP:conf/ijcai/Ceylan0R23} provided a parametrized complexity study of stable seat assignments when agents' preferences are cardinal.
It is also worth noting that in their investigation of computational aspects of hedonic seat arrangements, \citet{wilczynski2023ordinal} focused on simple structures similar to those that we study, such as cycles or paths.

Next, we mention some other settings that have connections with our problem. First, we note that allocating seats to agents is closely connected to the well-studied problem of allocating indivisible items~(see, e.g., \cite{manlove2013algorithmics}). In particular, within the vast literature on this problem, emphasis has been put on agents' envy and existence of blocking pairs (see, e.g.,  \citet{DBLP:conf/ijcai/MassandS19}).
Another important area related to exchange stability of allocations is the
\emph{Schelling segregation model} \cite{Schelling1969,Schelling1971}.  There, a number of agents, assigned different types, is placed on a graph. Their utility is dependent on the proportion of their neighbors of the same type. 
Then, the designer's goal is to ensure that no pair has an incentive to exchange their assignments. We note that, in contrast to the Schelling model, in our case the agents are not assigned types and the preferences depend only on the number of approved neighbors.
Notably, similarly to our approach, \citet{bilo2022topological} explored the case in which agents can only swap with their neighbors within the context of the Schelling model. 

Furthermore, the model we consider is also a special case of the \emph{land allocation with friends} problem studied by \citet{EPTZ20a}.
There, seats correspond to land plots. However, their paper focuses on the complexity of computing welfare maximizing allocations and on achieving truthfulness, rather than preventing exchanges of assigned plots.

Another such topic concerns one of the classical problems related to stability is the \emph{stable marriage problem}, in which men and women are matched to each other so that no man and woman pair would prefer to leave their matching to be with each other (see, e.g., \citep{GaSh62a}, \cite{irving1994stable}). In a variation of this problem, the \emph{stable roommate problem} (see \citep{irving1985efficient}), where any pair of agents can be matched together, \citep{cechlarova2002complexity}, as well as \citep{cechlarova2005exchange} considered the concept of exchange stability, showing that it is NP-complete to check if an exchange-stable matching exists.  We note that this fact directly implies hardness of checking whether a stable assignment exists in our setting.

\section{Preliminaries}

Let $\mathcal{A}=\{1,...,n\}$ be a set of $n$ agents which need to be assigned to vertices in an undirected \textit{seat graph} $G=(V,E)$ where edges $(v_i,v_j)$ indicate if seats $v_i$ and $v_j$ are adjacent.  
We assume that $|V|=n$. 
For convenience, we denote $[t] = \{1,2,\ldots,t\}$ for any positive integer $t$.
We consider agents with binary preferences, wherein each agent $i$ \emph{approves} of some subset of the other agents in $\mathcal{A}$.
This preference structure is expressed by a directed preference graph $\mathcal{P}$ on $\mathcal{A}$, where $(i,j)$ is an arc in $\mathcal{P}$ if and only if agent $i$ approves agent $j$.  
For notational convenience, for each pair of agents $i,j\in \mathcal{A}$, we write $i\rightarrow j$ if $(i,j)$ is an arc in $\mathcal{P}$. 
  It is often convenient to represent non-arcs explicitly in our model. So, if $(i,j)$ is not an arc in $\mathcal{P}$, we write $i\nrightarrow j$ (sometimes we also write $j \nleftarrow i$~).  Given a directed path $P$ in the preference graph, denote $\mathsf{head}(P)$ as the initial vertex and $\mathsf{tail}(P)$ as the terminal vertex of the path. 
Directed path $P$ is maximal if it cannot be extended while remaining a path.

	\begin{definition} 
		A \textit{path partition} of a directed graph $\mathcal{D}$ is a collection $\{P_j\}_{j=1}^k$, where each $P_i$ is a directed path and each vertex of $\mathcal{D}$ belongs to exactly one path. 
	\end{definition}
	A path partition $\{P_j\}_{j=1}^k$ is \textit{minimal} if $\mathsf{tail}({P_i}) \nrightarrow \mathsf{head}({P_j}) $ for any $i,j\in [k], i \neq j$. 
	Note that any path partition can be made minimal by simply combining paths whenever $\mathsf{tail}({P_i}) \rightarrow \mathsf{head}({P_j}) $ for some $i,j\in [k]$.

Let us now define an \textit{assignment}, which is a bijection $\pi:\mathcal{A}\rightarrow V$ assigning each agent to a vertex on the seat graph. 
For convenience, given an agent $i$ and an assignment $\pi$ we will denote as $N_{\pi}(i)$ the set of agents neighboring $i$ in $\pi$. 
For a given assignment $\pi$, the utility of an agent is as follows: 
$$u_i(\pi) = \sum\limits_{ j\in N_\pi(i)  } \mathbb{I}[i\rightarrow j] $$
where $\mathbb{I}$ is an indicator function.
In other words, agent $i$'s utility under $\pi$ is the number of $i$'s approved neighbors.

Denote $\pi^{i\leftrightarrow j}$ as the assignment $\pi$ with $\pi(i)$ and $\pi(j)$ switched. 
$\pi^{i\leftrightarrow j}$ represents the result of what we will call a \emph{swap} between $i$ and $j$.
We say an agent $i$ \textit{envies} agent $j$ if $u_i(\pi)<u_i(\pi^{i\leftrightarrow j})$. 
In other words, $i$ strictly improves her utility by exchanging seats with $j$. 
We say that agents $i$ and $j$ form a \emph{blocking pair} if both $i$ envies $j$ and $j$ envies $i$. Then, an assignment $\pi$ is \textit{stable} if it contains no blocking pairs. 
In this paper we are interested in the restricted variant of stability, in which agents are only allowed to swap if they are assigned to adjacent vertices.

\begin{definition}[Neighborhood Stability]
    An assignment $\pi$ is \emph{neighborhood stable} if there is no blocking pair assigned to adjacent vertices under $\pi$.
\end{definition}

Note that any blocking pair must first agree to swap, which requires that the two agents in the blocking pair interact.
Thus, neighborhood stability is well-motivated in any setting in which edges in the seat graph characterize potential interactions (e.g., the seating arrangement application). 

We now illustrate our setting by way of an example.

\begin{example}\label{ex:Initial}

Consider the instance shown in Figure \ref{fig:initialEx}.  
The preference graph shows that agents $a$ and $c$ approve each other, as do $b$ and $d$.
Moreover, $d$ approves $c$ while $b$ approves $a$.
The seat graph is a path with four vertices.

\begin{figure}[h]
  \centering
  \scalebox{0.8}{
     \begin{tikzpicture}
         [->,shorten >=1pt,auto,node distance=1.2cm, semithick]
         \node[shape=circle] (A)  {\Large $a$};
         \node[shape=circle, below of=A] (B)  {\Large $b$};
         \node[shape=circle, right of =A] (C)  {\Large $c$};
         \node[shape=circle, below of =C] (D)  {\Large $d$};
         \node[shape=circle, below of =D, yshift=1em] (D')  {};

         \draw [thick,->, bend left] (A) to  (C);
         \draw [thick,->, bend left] (B) to  (D);
         \draw [thick,->, bend left] (D) to  (B);
         \draw [thick,->, bend left] (C) to  (A);
         \draw [thick,->] (B) to  (A);
         \draw [thick,->] (D) to  (C);
         \draw [thick,->] (D) to  (C);
     \end{tikzpicture}}
    ~
    \hspace{4em}
  \scalebox{0.8}{
    \begin{tikzpicture}
        [->,shorten >=1pt,auto,node distance=1.2cm, semithick]
        \node[shape=circle,draw=black, fill=black] (A)  {};

        \node[shape=circle, above of=A, yshift=-1em] (A')  {\Large $a$};

        \node[shape=circle, draw=black, fill=black, right of=A] (B)  {};
        \node[shape=circle, above of=B, yshift=-1em] (B')  {\Large $d$};

        \node[shape=circle,draw=black, fill=black, right of =B] (C)  {};

        \node[shape=circle, above of=C, yshift=-1em] (C')  {\Large $b$};
        \node[shape=circle,draw=black, fill=black, right of =C] (D)  {};
        \node[shape=circle, above of =D, yshift=-1em] (D')  {\Large $c$};

        \node[shape=circle, below of =D, yshift=-0em] (D')  {};

        \draw [thick,-] (A) to  (C);
        \draw [thick,-] (B) to  (D);
        \draw [thick,-] (D) to  (C);
        \draw [thick,-] (C) to  (A);

    \end{tikzpicture}}
     \caption{Example of a preference graph and an assignment.}\label{fig:initialEx}
    \end{figure}
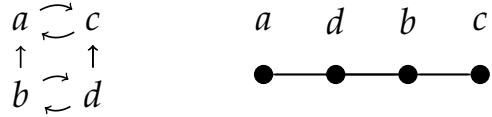

    Consider now the assignment $\pi$ shown in the right side of Figure \ref{fig:initialEx}. 
    Notice that $u_b(\pi) = u_d(\pi)=1$. 
    However, in the assignment $\pi^{b \leftrightarrow d}$, we have that both $b$ and $d$ improve their utility to two. 
    Hence, $b$ and $d$ form a blocking pair under $\pi$.
    Since $\pi(b)$ and $\pi(d)$ are adjacent in the seat graph, this means $\pi$ is not neighborhood stable.
    On the other hand, observe that $\pi^{b \leftrightarrow d}$ is stable, and thus also neighborhood stable.
\end{example}

\section{Non-Existence Results }
\label{sec:nonexistence}

As stable assignments are not guaranteed to exist, and we are studying a weakening of stability, a natural first question is whether neighborhood stable assignments always exist.  Toward this end, one promising observation is that   there are instances admitting neighborhood stable assignment for which every possible assignment is not stable, as shown by the following example.

\begin{example}\label{ex:adjacentharder}
  Take the seat graph consisting of two disjoint triangles, as depicted in Figure \ref{fig:adjacentharder}. 
  Also, let $\mathcal{A}=\{1,2,\ldots,6\}$, with the preference graph forming a directed cycle.
  As shown in, e.g., Example 2 in~\citet{DBLP:conf/ijcai/MassandS19}, this instance does not admit a stable assignment.
  To see why, note that any pair of agents assigned to distinct triangles who both have zero utility constitute a blocking pair, and in every assignment at least one such pair must exist.
  Nevertheless, one can check that every assignment in this instance is neighborhood stable. 
\begin{figure}[t]
  \centering
  \vspace{-0.1em}
    \scalebox{0.8}{
     \begin{tikzpicture}
         [->,shorten >=1pt,auto,node distance=1.2cm, semithick]
         \node[shape=circle] (A)  {\Large $a$};
         \node[shape=circle, right of =A] (B)  {\Large $b$};
         \node[shape=circle, below right of =B] (C)  {\Large $c$};
         \node[shape=circle, below left of =C] (D)  {\Large $d$};
         \node[shape=circle, left of =D] (E)  {\Large $e$};
         \node[shape=circle, above left of =E, xshift=-.3em] (F)  {\Large $f$};

         \draw [thick,->] (A) to (B);
         \draw [thick,->] (B) to (C);
         \draw [thick,->] (C) to (D);
         \draw [thick,->] (D) to (E);
         \draw [thick,->] (E) to (F);
         \draw [thick,->] (F) to (A);
     \end{tikzpicture}}
    ~
    \hspace{4em}
  \scalebox{0.8}{
    \begin{tikzpicture}
             [->,shorten >=1pt,auto,node distance=1.8cm,semithick]
             \node[shape=circle, draw=black, fill=black] (A)  {};
             \node[shape=circle, draw=black, fill=black, right of = A] (B)  {};
             \node[shape=circle, draw=black, fill=black, right of = A, above of = A, xshift=-2.5em] (C)  {};
             
             \draw [thick,-] (A) to  (B);
             \draw [thick,-] (B) to  (C);
             \draw [thick,-] (A) to  (C);

             [->,shorten >=1pt,auto,node distance=2cm, semithick]
             \node[shape=circle, draw=black, fill=black, right of = B] (D)  {};
             \node[shape=circle, draw=black, fill=black, right of = D] (E)  {};
             \node[shape=circle, draw=black, fill=black, right of = D, above of = D, xshift=-2.5em] (F)  {};
             \draw [thick,-] (D) to  (E);
             \draw [thick,-] (E) to  (F);
             \draw [thick,-] (D) to  (F);
    \end{tikzpicture}}
    \caption{Example of an instance for which any assignment fails stability, but satisfies neighborhood stability.}\label{fig:adjacentharder}
  \end{figure}
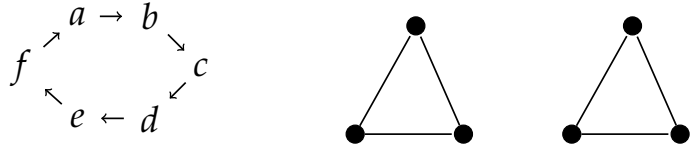
This holds because any blocking pair must be adjacent and thus belong to the same triangle, but a swap will not change either agent's neighborhood, leaving their utilities unchanged.
\end{example}

Despite the positive message of  Example~\ref{ex:adjacentharder}, in the following result, we show that this is not always the case. In particular, there are instances for which every assignment does not satisfy neighborhood stability. 
To show that, we give a general counterexample using balanced complete bipartite seat graphs to construct a family of instances on which neighborhood stable assignments do not exist. 
Figure \ref{fig:counterexample} depicts an instance from this family with only six agents.

\begin{theorem}\label{thm:negativeEx}
For every $n=2t$, where $t\geq 3$ is an odd integer, there is an instance with $n$ agents for which no neighborhood stable assignment exists.
\end{theorem}
\begin{proof}
    
    We consider an instance in which the seat graph is a complete bipartite graph with $t$ vertices in each partition, i.e., $K_{t,t}$.
    For convenience, we label the agents by $s_{ij}$ for each $i\in \{1,2\}$ and each $j\in [t]$. 
    Note that, for the purposes of this proof, $s_{i(t+1)}$ should be interpreted as $s_{i1}$ for $i\in \{1,2\}$.
    The preference graph $\mathcal{P}$ has the following arcs:
    \begin{itemize}
        \item $(s_{ij}, s_{kj})$ for all $j\in [t]$ and $i,k\in \{1,2\}$, $i\neq k$
        \item For each $i\in \{1,2\}$, $(s_{ij}, s_{i(j+1)})$ for all $j\in [t]$
    \end{itemize}
    Thus, $\mathcal{P}$ can be thought of as $t$ disjoint directed two-cycles and two disjoint directed $t$-cycles.
    The seat graph for six agents ($K_{3,3}$) and its associated preference graph, as described above, are pictured in Figure \ref{fig:negativeEx}.

    We point out that each agent $i$'s utility depends only on the set of agents assigned to the partition they are not assigned to.
    Thus, we can describe every assignment $\pi$ by the agents making up each partition, which we will denote $L$ and $R$. 
    We will also use $N'_\pi(i)=\mathcal{A}\setminus \{N_\pi(i)\cup \{i\}\}$ to denote the \emph{other} agents in $i$'s partition under $\pi$.
    We also make note of the fact that every vertex has an out-degree of two in $\mathcal{P}$ and thus each agent's maximum utility is two under any $\pi$. 
    Consider two sufficient conditions for an agent $i$ to envy agent $j$:

    \begin{enumerate}[{(C1)}]
        \item \label{cond:one} 
        \emph{Agent $i$ approves two agents from $N'_\pi(i)$ and $j\in N_\pi(i)$.}
        Agent $i$ envies agent $j$ because their out-degree in $\mathcal{P}$ is two, and thus they approve none of their neighbors under $\pi$. This means $u_i(\pi^{i\leftrightarrow j})=2 > 0 = u_i(\pi)$.
        \item \label{cond:two}
        \emph{Agent $i$ approves one of the agents in $N'_\pi(i)$, $j\in N_\pi(i)$, and $(i,j)$ is an arc in $\mathcal{P}$.} 
       	Agent $i$ envies agent $j$ since $|N_{\pi^{i\leftrightarrow j}}(i)| = |N'_\pi(i)\cup \{j\}| = 2 > 1 = u_i(\pi)$.
    \end{enumerate}
    
    We now show that, under any assignment, there always exists a pair of agents in opposite partitions who form a blocking pair.
    We begin with the case in which a partition contains a pair of agents forming a 2-cycle in the preference graph. 
    
    {\paragraph*{Case 1}: $s_{1j'}$ and $s_{2j'}$ are assigned to the same partition under $\pi$ for some $j'\in [t]$.}

\medskip

    First observe that each partition must contain an equal number of directed 2-cycles from $\mathcal{P}$.
    This fact holds since each directed 2-cycle which is not entirely contained in a partition is split evenly between the partitions, and thus the remaining directed 2-cycles must also be evenly distributed between partitions to ensure the partitions are of equal size.
    Since $t$ is odd, this implies the existence of a 2-cycle which is split between the partitions, i.e., $s_{1k}\in N_\pi(s_{2k})$ for some $k\in [t]$.

    Since there is at least one 2-cycle assigned to the same partition and there is at least one 2-cycle split between partitions, there exists an index $j\in [t]$ such that $s_{1j}$ and $s_{2j}$ are assigned to the same partition and $s_{1(j+1)}, s_{2(j+1)}$ are assigned to opposite partitions.
	Assume without loss of generality that $s_{1j}, s_{2j}, s_{1(j+1)}\in L$ and $s_{2(j+1)}\in R$.
    Note that, because $s_{2j}, s_{1(j+1)}\in N'_\pi(s_{1j})$, it stands that agent $s_{1j}$ envies every agent assigned to $R$ under $\pi$ due to (C\ref{cond:one}).

    We can assume henceforth that every agent $i$ assigned to $R$ under $\pi$ approves at least one agent in the opposite partition, i.e., $|N_\pi(i)|\geq 1$. 
    This holds since otherwise $i$ would envy $s_{1j}$ by (C\ref{cond:one}) and thus the two agents would form a blocking pair.
    Let $s_{1l}, s_{2l}$ be some 2-cycle in $\mathcal{P}$ assigned to $R$ under $\pi$ (which we know must exist as there is an equal number of 2-cycles assigned to each partition). 
    It is clear that both $s_{1(l+1)}$ and $s_{2(l+1)}$ must be assigned to $L$ under $\pi$ since otherwise $|N_\pi(s_{il})|=0$ for some $i\in \{1,2\}$. 
    In words, this means each 2-cycle assigned to $R$ under $\pi$ must be followed by a 2-cycle which is assigned to $L$ under $\pi$ as we increase our index.

    Since this holds for any 2-cycle assigned to $R$ under $\pi$, and each partition includes an equal number of 2-cycles, it must be that $s_{1(j-1)}, s_{2(j-1)}\in R$. But this means $s_{1(j-1)}$ envies $s_{1j}$ by (C\ref{cond:two}) since $s_{2(j-1)}\in N'_\pi(s_{1(j-1)})$ and $(s_{1(j-1)}, s_{1j})$ is an arc in $\mathcal{P}$. Thus, $s_{1(j-1)}$ and $s_{1j}$ form a blocking pair.

    {\paragraph*{Case 2}: $s_{1j}$ and $s_{2j}$ are assigned to opposite partitions under $\pi$ for all $j\in[t]$.}
    
    \medskip
    
    We will refer to agents of the form $s_{1j}$ as top agents and agents of the form $s_{2j}$ as bottom agents. 
    Since each agent is either top or bottom and a partition is composed of $t$ agents, it holds that in any partition there must be at least $\lceil t/2 \rceil$ agents from the same category.
    Also, since these agents belong to a directed $t$-cycle in $\mathcal{P}$, it must be that they induce at least one arc in $\mathcal{P}$. 
    This holds since the largest independent set of a $t$-cycle is of size $\lfloor t/2 \rfloor$ and $\lceil t/2 \rceil > \lfloor t/2 \rfloor$ since $t$ is odd.
    Assume, without loss of generality, that $s_{11}, s_{12}\in L$. 
    By our case assumption, $s_{21},s_{22}\in R$.
    It is immediate by (C\ref{cond:two}) that $s_{11}$ and $s_{21}$ envy each other.    
\end{proof}

\section{Neighborhood Stability on Cycles and Paths}
\label{sec:cycle_path}

Having shown that neighborhood stable assignments do not always exist, we turn next to the natural restricted cases studied extensively by \citet{berriaud2023stable} in the seat arrangement setting: cycle and path seat graphs.
A natural approach to proving existence of a neighborhood stable assignment is to use a potential function argument.
We begin by showing that this approach cannot work when the seat graph is a cycle.

We refer to \emph{swap dynamics} as a procedure which begins with an arbitrary assignment and allows adjacent blocking pairs to swap until no such pairs remain. 
We show next that this approach cannot work when the seat graph is a cycle. The proof is deferred to the the appendix.

	\begin{proposition} \label{prop: SwapDyn}
	There exists an instance and an assignment on that instance from which no swap dynamics converge, even when the seat graph is a cycle of length 4. 
	\end{proposition}

We note that even though the instance used to prove Proposition~\ref{prop: SwapDyn} does not converge to a neighborhood stable assignment through swap dynamics, it does admit a neighborhood stable assignment. In the next section we address the problem of whether this is the case for all instances where the seat graph is a cycle.\footnote{For example, consider any assignment under which agent $a$ neighbors $b$ and $c$, i.e., $N(a) = \{b,c \}$.} 

\subsection{Cycle Seat Graphs}
In what constitutes our central result, we will establish that a neighborhood stable assignment always exists when the seat graph is a cycle.
\begin{theorem}\label{thm: Cycle}
	A neighborhood stable assignment always exists and can be computed in polynomial time $O(n^3)$ when the seat graph is a cycle.
\end{theorem}

We point out that Theorem~\ref{thm: Cycle} is a surprising result given the findings of \citet{berriaud2023stable}. In particular, if we additionally consider blocking pairs distance two from each other, even in an instance with five agents, a stable assignment may not exist on a  cycle. Furthermore, as mentioned in the introduction, if we consider cardinal preferences, even with as few as four agents, a neighborhood stable assignment may not exist. 

We first give some intuition behind Algorithm~\ref{alg:cycle}, which is used to prove Theorem~\ref{thm: Cycle}.
	At a high level, our algorithm finds a minimal path partition and uses it to assign agents to the seat graph. If the resulting assignment is not neighborhood stable, we show that there is a different minimal path partition whose corresponding assignment has a strictly higher number of agents who approve an agent assigned immediately to their right. As the number of such agents is at most $n$ (which only happens when  the friendship graph contains a directed cycle as a subgraph), the algorithm is guaranteed to terminate on a neighborhood stable assignment.

We now define some useful notation. We will denote a seat graph with $n$ verties that is a cycle as $\mathcal{C}_n$. Further, given a path partition $\{P_j\}_{j=1}^k$, denote $\Phi(P_1,...,P_k)$ as the assignment that maps agents in $P_1$ to nodes $v_1,...,v_{|P_1|}$, agents in $P_2$ to nodes $v_{|P_1|+1},...,v_{|P_1|+|P_2|}$, and so forth. Refer to Figure~\ref{fig:ass} for an illustration of assignment $\Phi(P_1,\ldots,P_k)$.
Before proceeding with the proof of Theorem~\ref{thm: Cycle}, we make the following observations, which give sufficient conditions for blocking pairs.
\begin{figure}[t]
	\centering
	\includegraphics[scale=0.15]{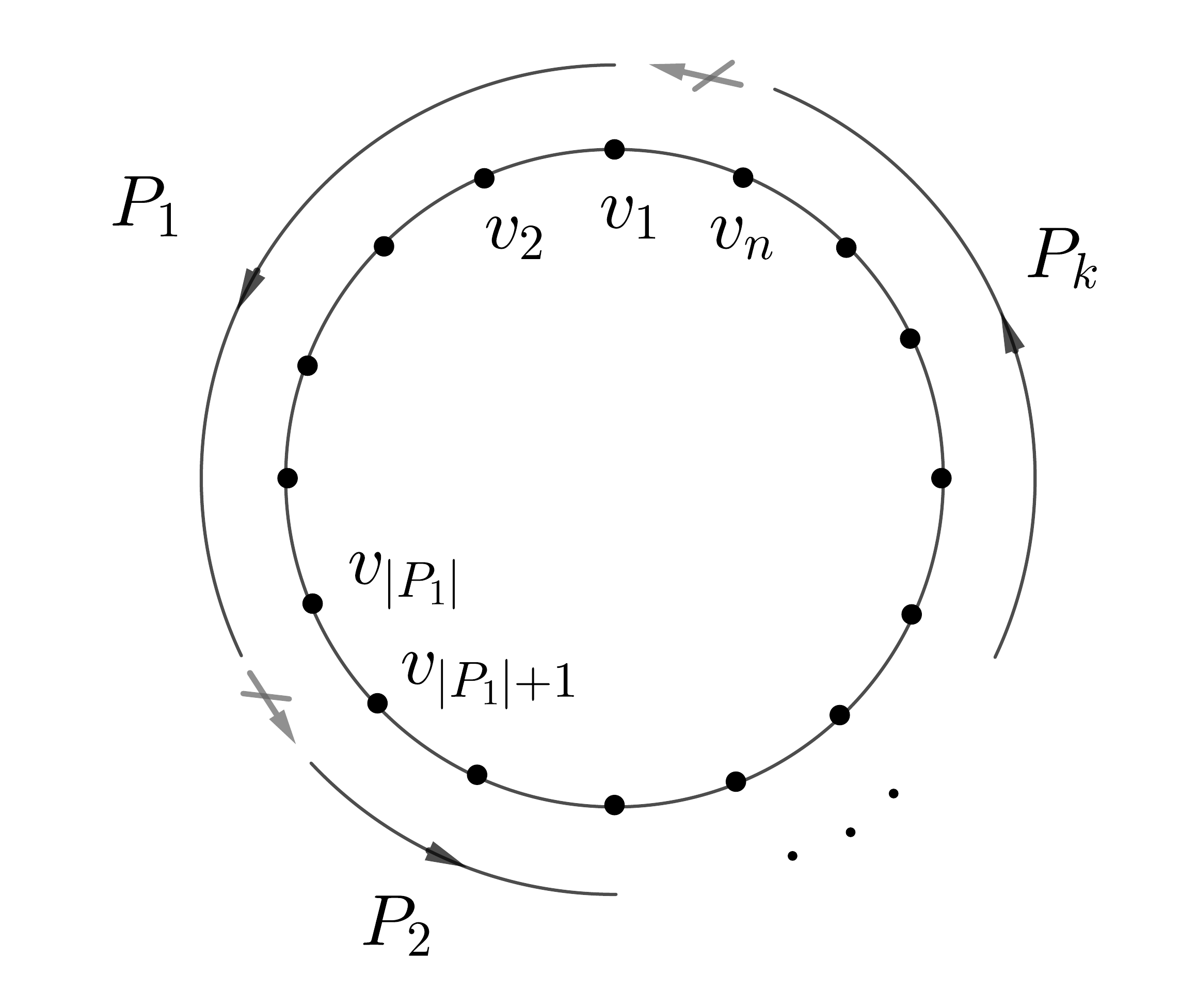}
	\caption{Illustration of $\Phi(P_1,\ldots,P_k)$, an assignment which maps a path partition to the vertices of the seat graph, in order. This assignment is used in line~\ref{alg:cycle:phi} of Algorithm~\ref{alg:cycle}.}\label{fig:ass}
\end{figure}

\begin{observation}\label{obs: cond1} 
	For every assignment $\pi$ on $\mathcal{C}_n$,  if a pair of agents assigned at adjacent nodes $(v_i,v_{i+1})$ satisfy either of the following two conditions:
		\begin{enumerate}[{(P1)}]
			\item $  \pi^{-1}(v_{i+1}) \nrightarrow \pi^{-1}(v_{i-1}) $ 
			\item $ \pi^{-1}(v_i) \nrightarrow \pi^{-1}(v_{i+2})$ 
	\end{enumerate}
then $\pi^{-1}(v_i)$  and $\pi^{-1}(v_{i+1})$ cannot form a blocking pair.
\end{observation} 

\begin{observation}\label{obs: cond2} 
	For any assignment $\pi$ on $\mathcal{C}_n$, if $\pi^{-1}(v_{i})$ approves $\pi^{-1}(v_{i-1})$ (resp. $\pi^{-1}(v_{i+1})$ ), then $\pi^{-1}(v_{i})$ does not envy $\pi^{-1}(v_{i+1})$ (resp. $\pi^{-1}(v_{i-1})$).  
\end{observation}

\begin{algorithm}[t]
	\caption{Neighborhood Stable Assignment on Cycles}
	\label{alg:cycle}
	\DontPrintSemicolon
	
	\KwIn{Preference graph $\mathcal{P}$ and seat graph $\mathcal{C}_n$}
	\KwOut{Neighborhood stable assignment  $\pi$}
	
	\BlankLine
	
	\If{$\exists s,t,w \in \mathcal{A}$ with $ s \rightarrow w \leftarrow t, s \nleftarrow t,s \nrightarrow t $}
	{Set $\pi(s)=v_1,\pi(w)=v_2,\pi(t)=v_3$ \;
		\For{$i=4$ to $n$} 
		{
			Set $W^{i-1} =  \bigcup_{r=1}^{i-1} \pi^{-1}(v_r) $\;
			Let $Q= \{\ell \in \mathcal{A}\setminus W^{i-1} \ | \ \ell \rightarrow \pi^{-1}(v_{i-1}) \}$   \;
			\If{$Q\neq \emptyset$}
			{
				Set $\pi^{-1}(v_{i})  =q$ where $q \in Q$ 
			}
			\Else
			{
				Let $\pi^{-1}(v_{i})$ be any element in $\mathcal{A}\setminus W^{i-1}$ 
			}
		}
		
	}
	\Else
	{
		
		Let $\{P_j\}_{j=1}^k$ be a minimal path partition of $\mathcal{P}$ \;
		
		Set $\pi=\Phi(P_1,...,P_k)$ \label{alg:cycle:phi}
		
		\While{$\pi$ contains an adjacent blocking pair }{
			Update $\pi$ according to Lemma~\ref{lem: reass}
		} 
		
	}
	\Return{$\pi$}
\end{algorithm}

\vspace{1em}
\noindent\emph{Proof of Theorem~\ref{thm: Cycle}.}
We begin our proof focusing on a particular preference structure that admits a simple argument.

{\paragraph*{Case 1: }$\exists s,t,w \in \mathcal{A}$ with $s \rightarrow w \leftarrow t,s \nleftarrow t,s \nrightarrow t $ }

\medskip

We first set $\pi(s) = v_1$, $\pi(w) = v_2$, and $\pi(t) = v_3$, and sequentially define an assignment for $4\leq i \leq n$ as described by the for loop in Algorithm~\ref{alg:cycle}. We now prove that $\pi$ is neighborhood stable.

Observe that $s$ does not envy $\pi^{-1}(v_n)$ and   $t$ does not envy $\pi^{-1}(v_4)$ by Observation~\ref{obs: cond2}. Note that the agent pairs $s,w$ and $w,t$ cannot form  blocking pairs by Observation~\ref{obs: cond1}. Thus, we only need to check whether there is an adjacent blocking pair in $\{\pi^{-1}(v_4), ..., \pi^{-1}(v_n)\}$. Take any agent pair $\pi^{-1}(v_i),\pi^{-1}(v_{i+1})$ and observe that  one of the following two must hold:
\begin{itemize}
		\item $ \pi^{-1}(v_i) \rightarrow \pi^{-1}(v_{i-1}) $
		\item $ \pi^{-1}(v_{i+1}) \nrightarrow \pi^{-1}(v_{i-1}) $
\end{itemize}

To see this, note that if $\pi^{-1}(v_i) \nrightarrow \pi^{-1}(v_{i-1})$, then it must hold that  $ \pi^{-1}(v_{i+1}) \nrightarrow \pi^{-1}(v_{i-1}) $. This is because $ \pi^{-1}(v_{i+1}) \in \mathcal A\setminus W^{i-1}$ and if $  \pi^{-1}(v_{i+1}) \rightarrow \pi^{-1}(v_{i-1})$, then the agent $\pi^{-1}(v_{i+1})$ would have been selected in place of $\pi^{-1}(v_i)$.

In either  case ($\pi^{-1}(v_i) \rightarrow \pi^{-1}(v_{i-1}) $ or $  \pi^{-1}(v_{i+1}) \nrightarrow \pi^{-1}(v_{i-1})$), the agent pair $\pi^{-1}(v_i),\pi^{-1}(v_{i+1})$ cannot form a blocking pair by Observation~\ref{obs: cond2} and Observation~\ref{obs: cond1}, respectively. Therefore, $\pi$ is neighborhood stable. 

{\paragraph*{Case 2}: For each $w\in\mathcal{A}$, and for any $s,t\in N^- ({w,\mathcal A})$ either $s\rightarrow t$ or $t \rightarrow s$ or both } 

\medskip

We start by constructing a path partition as follows: compute a maximal path $P_1$ in $\mathcal{P}$, remove the vertices of $P_1$, and then compute maximal paths in the subgraph of $\mathcal{P}$ induced by the remaining vertices. This procedure is iterated until no vertices are left, at which point we are left with the path partition $\{P_j\}_{j=1}^k$. Note that for this assignment $\mathsf{tail}({P_i}) \nrightarrow \mathsf{head}({P_j}) $ for any $i,j\in [k]$ as otherwise this would contradict maximality of $P_i$ or $P_j$ in a given iteration. Hence $\{P_j\}_{j=1}^k$ is a minimal path partition.

\begin{claimN}\label{lem:pathext} For instances satisfying Case 2, given a directed path $P=s_1 \rightarrow s_2 \rightarrow \cdots \rightarrow s_r$, if there exists an agent $s\not\in P$  with $s\rightarrow s_r$ then either $s\rightarrow s_1 $  or there exists an index $j\in [r-1]$ such that $s_j\rightarrow s \rightarrow s_{j+1}$ .
\end{claimN} 
\begin{proof}[Proof of Claim] First, note that if $|P|=1$, we can trivially extend $P$ since $s\rightarrow s_1$ by assumption. We proceed by induction on the size of the directed path. Suppose the statement is true for $|P|= r-1$. Now, we prove that the statement holds for all directed paths of length $r$. Let $P$ be $s_1 \rightarrow s_2 \rightarrow \cdots \rightarrow s_r$, and consider the subpath $P'=s_2 \rightarrow \cdots \rightarrow s_r$ of length $r-1$.
	
	By the induction hypothesis, either $s\rightarrow s_2$, or there is an index $2\leq j\leq r-1$ such that $s_j \rightarrow s \rightarrow s_{j+1} $. In the latter case, the statement immediately holds.
	
	Consider the former scenario, $s\rightarrow s_2 $. Since $s_1\rightarrow s_2$, we have that $s_{1},s\in N^{-}(s_2,\mathcal {A})$. By the case distinction, there must be an arc $s_{1}\rightarrow s$ or $s \rightarrow s_{1}$. In the case of $s_{1}\rightarrow s$, we observe that $s_{1}\rightarrow s\rightarrow s_2 $. Similarly, if $s \rightarrow s_{1}$, then $s\rightarrow s_1 \rightarrow s_2$. Consequently, the induction step is satisfied.
\end{proof}

We now present the main technical lemma, which states that under Case 2, if there is an adjacent blocking pair, then there exists another assignment with a strictly higher number of agents who approve their right neighbor. 

\begin{lemma}\label{lem: reass}
	Let $\pi$ be an assignment $\Phi(P_1,\ldots,P_k)$ corresponding to a minimal path partition $\{P_j\}_{j=1}^k$. If there is an adjacent blocking pair for $\pi$, then there is another minimal path partition $\{P'_j\}_{j=1}^{k'}$ whose corresponding  assignment $\rho=\Phi(P'_1,\ldots,P'_{k'})$ satisfies 
	\begin{align*}
			|\{v_j\in V \ | \ \rho^{-1}(v_j)\rightarrow \rho^{-1}(v_{j+1}) \ \} | \geq  |\{v_j\in V \ | \ \pi^{-1}(v_j)\rightarrow \pi^{-1}(v_{j+1}) \ \} | +1.
	\end{align*}

\end{lemma}

\begin{proof}
	Suppose $\pi^{-1}(v_i), \pi^{-1}(v_{i+1})$ form a blocking pair. First, observe that adjacent blocking pairs come in only two possible types: either they belong to the same directed path in the path partition, or $\pi^{-1}(v_i) \in P_{r}$ and $\pi^{-1}(v_{i+1}) \in P_{r+1}$ for some $r \in [k]$.\footnote{For notational convenience, we interpret $P_{k+1}$ as $P_{1}$.} 
	In what follows, we treat these cases separately.
	{\paragraph*{Type I}: $ \pi^{-1}(v_i), \pi^{-1}(v_{i+1})\in P_{r} $ for some $r\in [k]$  }
	
	\medskip
	
	By Observation~\ref{obs: cond2}, it must be the case that $\pi^{-1}(v_{i+1}) \nrightarrow \pi^{-1}(v_{i+2})$ since $\pi^{-1}(v_i), \pi^{-1}(v_{i+1})$ form a blocking pair. It follows that, $\pi^{-1}(v_{i+1})=\mathsf{tail}(P_r)$ and $\pi^{-1}(v_{i+2})=\mathsf{head}(P_{r+1})$. 
	Since $\pi^{-1}(v_i), \pi^{-1}(v_{i+1})$ form a blocking pair, we know that $\pi^{-1}(v_{i})\rightarrow \pi^{-1}(v_{i+2}) $ and $\pi^{-1}(v_{i+1})\rightarrow \pi^{-1}(v_{i-1})$ by Observation~\ref{obs: cond1}.
	
	Consider the case when $|P_r|=2$ which means $\pi^{-1}(v_i)=\mathsf{head}(P_{r})$  and $\pi^{-1}(v_{i-1})=\mathsf{tail}(P_{r-1})$.  
	Recalling that $\pi^{-1}(v_{i+1})\rightarrow \pi^{-1}(v_{i-1})$ and applying Claim~\ref{lem:pathext}, we see that $\pi^{-1}(v_{i+1})$ can be inserted into $P_{r-1}$ to form a new path $P'_{r-1}$. 
	Also, since $\pi^{-1}(v_{i})\rightarrow \pi^{-1}(v_{i+2}) = \mathsf{head}(P_{r+1})$, we can construct a new directed path where $P'_{r+1}=\pi^{-1}(v_{i}) \rightarrow P_{r+1}$. 
	Consider a new path partition $P'=\{P_1,\ldots,P_{r-2},P'_{r-1},P'_{r+1},P_{r+2},\ldots, P_k\}$. 
	We can make this path partition minimal by connecting tails of the new directed paths to the heads of the new directed path whenever possible. 
	Consider now an assignment $\rho=\Phi(P'_1,...,P'_{k'})$ corresponding to the new minimal path partition $\{P'_j\}_{j=1}^{k'}$. 
	Noting that $k'\leq k-1$, we see that

		\begin{align*}
			|\{v_j\in V \ | \ \rho^{-1}(v_j)\rightarrow \rho^{-1}(v_{j+1}) \ \} | &\geq n - k' \\ 
			& \geq n-(k-1) \\
			&= |\{v_j\in V \ | \ \pi^{-1}(v_j)\rightarrow \pi^{-1}(v_{j+1}) \ \} | +1.
		\end{align*}

	Now consider the case where $|P_r|\geq 3$. 
	Recall that $\pi^{-1}(v_{i+1})=\mathsf{tail}(P_r)$ and $\pi^{-1}(v_{i+2})=\mathsf{head}(P_{r+1})$.
	Also, because $|P_r|\geq 3$, we have that $\pi^{-1}(v_{i-1})\in P_r$. 
	Since $\pi^{-1}(v_{i+1})\rightarrow \pi^{-1}(v_{i-1})$, it follows from Claim~\ref{lem:pathext} that the directed path $P_r\setminus \{ \pi^{-1}(v_i), \pi^{-1}(v_{i+1}) \}$ can be extended by inserting $\pi^{-1}(v_{i+1})$. 
	This gives us a new path $\tilde{P}$ with $\mathsf{tail}(\tilde{P})=\pi^{-1}(v_{i-1})$ and with $\pi^{-1}(v_{i+1})\in \tilde{P}$. 
	
	Now note that $\pi^{-1}(v_{i})\rightarrow \pi^{-1}(v_{i+2})=\mathsf{head}(P_{r+1})$. 
	Hence, we can combine directed paths $P_r$ and $P_{r+1}$ to a single directed path $P'_{r+1}=\tilde{P}\rightarrow \pi^{-1}(v_{i}) \rightarrow P_{r+1} $.
	This process of combining two paths is illustrated in Figure~\ref{fig:merge}.
	Observe that $P'_{r+1}$ is indeed a directed path since we have $\mathsf{tail}(\tilde{P})=\pi^{-1}(v_{i-1})\rightarrow \pi^{-1}(v_{i})  \rightarrow \pi^{-1}(v_{i+2})=\mathsf{head}(P_{r+1}) $. 
	Finally, consider a new path partition $P'=\{P_1,\ldots, P_{r-1}, P'_{r+1},\ldots,P_k\}.$ 
	We make this path partition minimal by connecting tails of the new directed paths to the heads of the new directed paths whenever possible. 
	Consider now an assignment $\rho=\Phi(P'_1,...,P'_{k'})$ corresponding to the new minimal path partition $\{P'_j\}_{j=1}^{k'}$.
	Noting that $k'\leq k-1$, we can show by the same argument that the inequality holds.

	\begin{figure}
		\centering
		\includegraphics[scale=0.2]{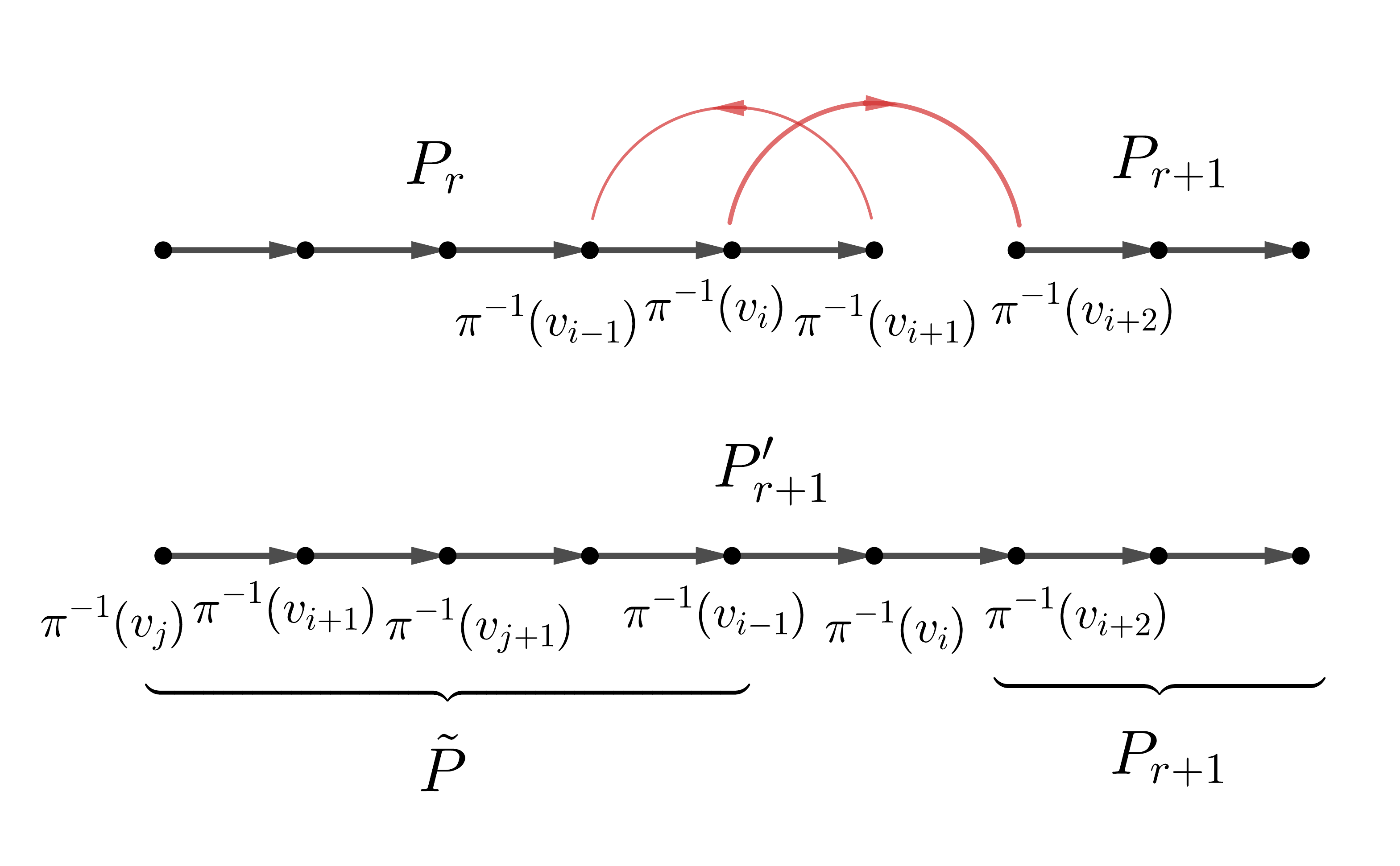}
		\caption{Demonstration of how directed paths $P_r$ and $P_{r+1}$ are combined to create a new path $P'_{r+1}$ when $|P_r| \geq 3$ and $\pi^{-1}(v_i),\pi^{-1}(v_{i+1})$ form a blocking pair of Type I.}\label{fig:merge}
	\end{figure}

	{\paragraph*{Type II:}   $\pi^{-1}  (v_i)\in P_{r}, {\pi}^{-1}(v_{i+1})\in P_{r+1}$ for some $r\in [k]$  }
	
	\medskip
	
	Note that Type II implies $\pi^{-1}(v_i) = \mathsf{tail}(P_r)$ and $\pi^{-1}(v_{i+1}) = \mathsf{head}(P_{r+1})$.
	Observe that for Type II to occur, it must be the case that $|P_{r+1}|=1$. To see this, suppose, on the contrary, that $|P_{r+1}|\geq 2$. It follows that $\pi^{-1}(v_{i+1}) \rightarrow \pi^{-1}(v_{i+2})$ and thus, by Observation~\ref{obs: cond2}, $\pi^{-1}(v_{i+1})$ does not envy $\pi^{-1}(v_i)$.
	This contradicts our initial assumption that $\pi^{-1}(v_i), \pi^{-1}(v_{i+1})$ form a blocking pair.
	
	Since we have that $|P_{r+1}|=1$, it follows that $\pi^{-1}(v_{i+2})=\mathsf{head}(P_{r+2})$. 
	But since $\pi^{-1}(v_{i})=\mathsf{tail}(P_r)$ and $\pi^{-1}(v_{i}) \rightarrow \pi^{-1}(v_{i+2})$, we see that $\mathsf{tail}(P_r)\rightarrow \mathsf{head}(P_{r+2})$. This contradicts the minimality of the path partition. 
\end{proof}

If we repeatedly apply Lemma~\ref{lem: reass} to the assignment $\pi$ corresponding to the initial minimal path partition $\{P_j\}_{j=1}^k$, we are guaranteed to terminate with a neighborhood stable assignment. 
To see this, observe that for any assignment $\rho$, we have $|\{v_j\in V \ | \ \rho^{-1}(v_j)\rightarrow \rho^{-1}(v_{j+1}) \ \} |\leq n$, and thus Lemma~\ref{lem: reass} is applied at most $n$ times.
The assignment resulting from the last application of Lemma~\ref{lem: reass} cannot have an adjacent blocking pair (otherwise Lemma~\ref{lem: reass} could be applied again) and is thus neighborhood stable.

\vspace{1em}
\noindent\textbf{Analysis of running time}\\ 
Lastly, we analyze the running time of Algorithm~\ref{alg:cycle}. 
Note that the first step, checking which of the two cases the preference graph $\mathcal P$ belongs to, can be done in $O(n^3)$ time by examining every triplet. 
In Case 1, the assignment can be computed in at most $O(n^2)$ time.
In Case 2, the initial minimal path partition can be found in $O(n^2)$ time and finding assignment $\rho$ in Lemma~\ref{lem: reass} takes $O(n^2)$ time. Since Lemma~\ref{lem: reass} is applied at most $n$ times, the total running time is $O(n^3)$.
\qed

\subsection{Path Seat Graphs}
Although computation of a neighborhood stable assignment on a path can be achieved by a straightforward reduction from our algorithm for the cycle case (Algorithm~\ref{alg:cycle}), we will give a polynomial-time algorithm with stronger guarantees.
Specifically, our algorithm outputs an assignment with no blocking pairs at distance less than or equal to two from each other.
This answers a question left by \citet{berriaud2023stable}, who gave a potential function argument for the same result on paths, but left open whether such an assignment could be computed with an explicit, polynomial-time algorithm.

\begin{theorem}
	An assignment in which no blocking pair are assigned at most distance two from each other always exists and can be computed in polynomial time $O(n^3)$ when the seat graph is a path.
\end{theorem}
\begin{proof}
	We first set $\pi(v_1)$ as an arbitrary agent $s\in \mathcal{A}$ and sequentially define an assignment for $k\geq 2$ as follows. 
	First, set  $W^{k-1} =  \bigcup_{i=1}^{k-1} \pi^{-1}(v_i). $ Then, let $Q=\{\ell \in \mathcal{A} \setminus W^{k-1} \ | \ \ell \rightarrow \pi^{-1}(v_{k-1})\}$. If $Q \neq \emptyset$, we assign an arbitrary agent from $Q$ to vertex $v_{k}$. On the other hand, if $Q = \emptyset$, we assign any agent from the set  $\mathcal{A} \setminus W^{k-1}$ to vertex $v_{k}$. Let $\pi$ be the resulting assignment.
	
	We will now rule out blocking pairs of the form $\pi^{-1}(v_{i}),\pi^{-1}(v_{i+1})$ and $\pi^{-1}(v_{i}),\pi^{-1}(v_{i+2})$ for all $i\in[n-2]$. When $i=1$, it is clear that neither  $\pi^{-1}(v_{2})$  nor $\pi^{-1}(v_{3})$ envies $\pi^{-1}(v_{1})$ since their neighborhood under a swap is a subset of their neighborhood under $\pi$. 
	
	Otherwise, for any $ i\geq 2 $, there are two cases. First, suppose $ \pi^{-1}(v_{i}) \rightarrow \pi^{-1}(v_{i-1})$. Assume for a contradiction that either $\pi^{-1}(v_{i}),\pi^{-1}(v_{i+1})$ or $\pi^{-1}(v_{i}),\pi^{-1}(v_{i+2})$ forms a blocking pair. This means that  $\pi^{-1}(v_{i})$ gets utility of two by swapping with either $\pi^{-1}(v_{i+1})$ or $\pi^{-1}(v_{i+2})$. Since $\pi^{-1}(v_{i+1})$ is a neighbor under either such swap, it must be the case that $\pi^{-1}(v_{i})$ approves $\pi^{-1}(v_{i+1})$, and thus already has utility of two under $\pi$. This is a contradiction since the maximum achievable utility is two.
	
	Now suppose $ \pi^{-1}(v_{i}) \nrightarrow \pi^{-1}(v_{i-1})$. By the construction of the assignment, neither $\pi^{-1}(v_{i+1})$ nor $\pi^{-1}(v_{i+2})$ approves $\pi^{-1}(v_{i-1})$. We can conclude that neither   $\pi^{-1}(v_{i+1})$ nor $\pi^{-1}(v_{i+2})$ envies $\pi^{-1}(v_{i})$. 
	This holds because they do not approve their only new neighbor (i.e., $\pi^{-1}(v_{i-1})$) under a swap with  $\pi^{-1}(v_{i})$.
\end{proof}

\section{A Sufficient Condition for General Graphs }

In this section, we give a quite general, sufficient condition under which neighborhood stable assignments are guaranteed to exist.
Our condition constrains the relation between a measure of acyclicity on the preference graph and the number of \emph{leaves} (i.e., degree-one vertices) in the seat graph.
The directed feedback vertex set number, a well-studied directed graph property which we define here in the context of our problem, gives a measure of how close a directed graph is to being acyclic.

\begin{definition}
	The set $X\subseteq \mathcal{A}$ is a directed feedback vertex set (DFVS) if the subgraph of $\mathcal{P}$ induced by $\mathcal{A}\setminus X$ is acyclic. 
	The  DFVS number is the size of the smallest such set. 
\end{definition}

In what follows, we will prove the existence of a neighborhood stable assignment for every instance in which the DFVS number is upper bounded by the number of leaves in the seat graph.
The proof yields a polynomial-time algorithm for a stable assignment in the case that the DFVS number is zero, and $\mathcal{P}$ is thus a DAG. 
The proof is similar to that of Theorem 4.2 from \citet{bullinger2024topological}, however our statement captures a broader class of instances and concerns a different form of stability.

\begin{theorem} \label{thm:sufficient}
	For any instance in which the number of leaves of the seat graph $G$ exceeds the DFVS number $\gamma$ of the preference graph $\mathcal{P}$, there exists a neighborhood stable assignment on $G$.
	When $\gamma=0$, i.e., $\mathcal{P}$ is a DAG, a stable assignment can be computed in polynomial time.
\end{theorem}
\begin{proof}
	Let $X$ be the smallest DFVS in $\mathcal{P}$, and $L$ be the set of leaf nodes in the seat graph $G$.  We first assign agents in $X$ to vertices in $L$. Since $|X|\leq |L|$ by the assumption of the theorem such a partial assignment is possible. 
	
	Note that the remaining agents $\mathcal{A}\setminus X$ induce an acyclic subgraph of $\mathcal{P}$. We denote $\mathtt{sink}(\mathcal{D})$ as the sink node in an acyclic preference graph $\mathcal{D}$. We also let $\mathcal{P}[S]$ be the subgraph induced in $\mathcal{P}$ by $S\subseteq \mathcal{A}$.
	
	We will now describe the assignments of the remaining agents who have not yet been assigned seats. Let $s_1 = \mathtt{sink}(\mathcal{P}[\mathcal{A} \setminus X])$, and recursively set $s_{\ell+1}=\mathtt{sink}\left(\mathcal{P}\left[\mathcal{A}\setminus \left( X \cup \bigcup_{i=1}^\ell s_i\right)\right]\right)$. We now let agents $s_1,..., s_{n-\gamma}$ pick their favorite seats in order, breaking ties arbitrarily. Let $\pi$ denote the corresponding assignment, and for  $S\subseteq \mathcal{A}$ we denote $\pi(S)$ as the set of vertices that agents in $S$ are assigned to under $\pi$. 
	
	Consider any pair of agents $s_i,s_j\in \mathcal{ A} \setminus X$ 
	with $i<j$. We now argue that $s_i$ does not envy $s_j$. 
	The vertex $\pi(s_j)$  was 
	available when it was agent $s_i$'s turn to choose. Since  $s_{i}=\mathtt{sink}
	\left(\mathcal{P}\left[\mathcal{A}\setminus \left( X 
	\cup \bigcup_{r=1}^{i-1} s_r\right)\right]\right)$,  
	we have that the set of agents $H_i$ that $s_i$ 
	approves  satisfies $H_i \subseteq  X \cup\{s_1,...,s_{i-1}\}$ and 
	$$\pi(s_i)=\argmax\limits_{v\in  V\setminus \pi(X \cup\{s_1,...,s_{i-1}\})}  |N(v)\cap \pi(H_i)| .$$
	Finally, it follows that
	\begin{align*}
		u_{s_i}(\pi)&=|N(\pi(s_i))\cap \pi(H_i)| \\
		&\geq |N(\pi(s_j))\cap \pi(H_i)| \\
		&= u_{s_i}(\pi^{ s_i\leftrightarrow s_j}) .
	\end{align*}
	
	This means that any blocking pair must contain at least one agent from $X$. When $\gamma =0$ (i.e., $\mathcal{P}$ is a DAG), we have that $X=\emptyset$ and thus there are no blocking pairs, regardless of the distance between agents. Since we can check for acyclicity in directed graphs in polynomial time, it follows that our algorithm yields a stable assignment in polynomial time.     
	
	Suppose $x_i,x_j$  form a blocking pair that is adjacent in the seat graph under $\pi$. At least one of the two must be in $X$, say $x_i$. Since $\pi(x_i)$ is a leaf in $G$, we have that $x_j$ cannot envy $x_i$ since by swapping, $x_j$'s only neighbor becomes $x_i$. 
\end{proof}

As an example of how Theorem~\ref{thm:sufficient} can be applied, consider planar preference graphs.
Since planar directed graphs have a DFVS number of at most $\frac35 n$ \citep{Bo79}, Theorem~\ref{thm:sufficient} implies that any seat graph where at least $\frac35$ of the nodes are leaves will admit a neighborhood stable assignment.
An example of such a seat graph is a full, $m$-ary tree with $m\geq 3$, which is guaranteed to have at least $\frac{m-1}{m}n$ leaves.
Though we have given a very specific application of Theorem~\ref{thm:sufficient}, we note that the theorem provides a general condition that can be applied across broad graph classes.

\section{Conclusion} 

In this paper, we have initiated the study of neighborhood stability in problems where we assign agents to vertices of graphs.
We prove that for cycle and path seat graphs, a neighborhood stable assignment can be computed in polynomial time.
En route, we demonstrate a connection between path partitions and neighborhood stable assignments.

A natural next step is to identify other seat graph restrictions for which existence of a neighborhood stable assignment is guaranteed for all preference graphs, such as seat graphs which are grids or trees.
More generally, results in the style of Theorem~\ref{thm:sufficient} will help illuminate the structure of instances that always admit a neighborhood stable assignment.

Another interesting direction is to establish how neighborhood stability affects efficiency concepts.
In this vein, future work could bound the price of neighborhood stability, or alternatively, compute outcomes which maximize social welfare subject to neighborhood stability.
Lastly, while the results in this paper pertain to binary preferences, it remains to be seen whether they extend to more general preferences.
In the case of the cycle, an impossibility arises when extending to non-negative cardinal preferences, but this question remains open for the path seat graph.

\section*{Acknowledgments}

This project has received funding from the European 
    Research Council (ERC) under the European Union’s Horizon 2020 
    research and innovation programme (grant agreement No 101002854). Mashbat Suzuki is supported by the ARC Laureate Project FL200100204 on ``Trustworthy AI''.
    
	\begin{center}
    \noindent \includegraphics[width=3cm]{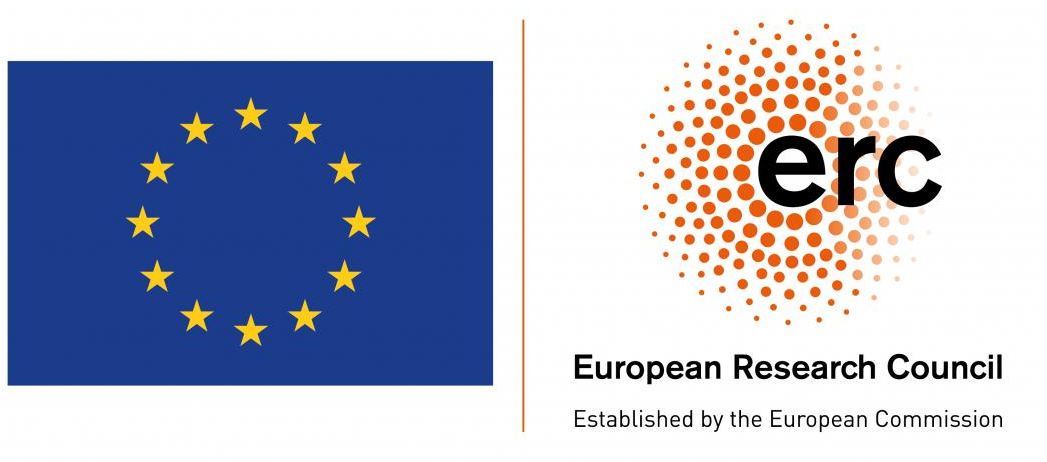}

	\end{center}

\bibliographystyle{plainnat}
\bibliography{sample.bib}

\appendix

\section{Proofs from Section~\ref{sec:cycle_path}}

	\textbf{Proof of Proposition \ref{prop: SwapDyn}}
\begin{proof}
	Consider an instance with four agents $\{a,b,c,d\}$ and a seat graph and preference graph as depicted in Figure \ref{fig:oscilation}. 
	Let us now consider the assignment $\pi$ shown in the right side of Figure \ref{fig:oscilation} and demonstrate that swap dynamics cannot converge starting from $\pi$.
	
	\begin{figure}[h]
		\centering
		\scalebox{0.7}{
			\begin{tikzpicture}
				[->,shorten >=1pt,auto,node distance=1.5cm, semithick]
				\node[shape=circle] (A)  {\huge a};
				\node[shape=circle, below of=A] (B)  {\huge b};
				\node[shape=circle, right of =A] (C)  {\huge c};
				\node[shape=circle, below of =C] (D)  {\huge d};
				\node[shape=circle, below of =D, yshift=1em] (D')  {};
				
				\draw [thick,->] (A) to  (B);
				\draw [thick,->] (B) to  (D);
				\draw [thick,->] (D) to  (C);
				\draw [thick,->] (C) to  (A);
				\draw [thick,->] (A) to  (D);
		\end{tikzpicture}}
		~
		\hspace{4em}
		\vspace{-0em}
		\scalebox{0.7}{
			\begin{tikzpicture}
				[->,shorten >=1pt,auto,node distance=1.5cm, semithick]
				\node[shape=circle,draw=black, fill=black] (A)  {};
				
				\node[shape=circle, above of=A, yshift=-3em] (A')  {\Large a};
				
				\node[shape=circle, draw=black, fill=black, below of=A] (B)  {};
				\node[shape=circle, below of=B, yshift=2.7em] (B')  {\Large d};
				
				\node[shape=circle,draw=black, fill=black, right of =A] (C)  {};
				
				\node[shape=circle, above of=C, yshift=-3em] (C')  {\Large b};
				\node[shape=circle,draw=black, fill=black, below of =C] (D)  {};
				\node[shape=circle, below of =D, yshift=2.6em] (D')  {\Large c};
				
				\draw [thick,-] (A) to  (B);
				\draw [thick,-] (B) to  (D);
				\draw [thick,-] (D) to  (C);
				\draw [thick,-] (C) to  (A);
				
		\end{tikzpicture}}
		\caption{Example of a preference graph and an assignment from which swap dynamics never converge.}\label{fig:oscilation}
	\end{figure}
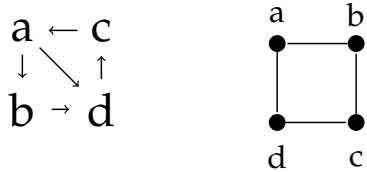
	
	First, observe that $N_{\pi}(c)=\{b,d\}$, while $N_{\pi}(b)=\{c,a\}$. Let now $\pi_1 = \pi^{b \leftrightarrow c}$. Then, $N_{\pi_1}(c)=\{b,a\}$, while $N_{\pi_1}(b)=\{c,d\}$. It follows that $u_c(\pi) = 0 $ while  $u_c(\pi_1) =1 $, and $u_b(\pi) = 0 $ while  $u_b(\pi_1) =1 $. So, $b$ and $c$ form a blocking pair. 
	
	Now, let $\pi_2 = \pi_1^{a \leftrightarrow d}$. Notice now that $N_{\pi_1}(a)=\{c,d\}$, while $N_{\pi_1}(d)=\{a,b\}$.  Also, $N_{\pi_2}(a)=\{b,d\}$, while $N_{\pi_2}(d)=\{a,c\}$. It follows that $u_a(\pi_1) = 1 $ while  $u_a(\pi_2) =2 $, and $u_d(\pi_1) = 0 $ while  $u_d(\pi_2) =1 $. So, $a$ and $d$ form a blocking pair. 
	
	Now, let $\pi_3 = \pi_2^{b \leftrightarrow c}$. Notice now that $N_{\pi_2}(c)=\{b,d\}$, while $N_{\pi_3}(c)=\{b,a\}$.  Also, $N_{\pi_2}(b)=\{c,a\}$, while $N_{\pi_3}(b)=\{c,d\}$. It follows that $u_c(\pi_2) = 0 $ while  $u_c(\pi_3) =1 $, and $u_b(\pi_2) = 0 $ while  $u_b(\pi_3) =1 $. So, $b$ and $c$ form a blocking pair. 
	
	Observe now that $\pi_3^{a \leftrightarrow d} = \pi $. Notice now that $N_{\pi_3}(a)=\{c,d\}$, while $N_{\pi}(a)=\{b,d\}$.  Also, $N_{\pi_3}(d)=\{a,b\}$, while $N_{\pi}(d)=\{a,c\}$. It follows that $u_a(\pi_3) = 1 $ while  $u_a(\pi) =2 $, and $u_d(\pi_3) = 0 $ while  $u_d(\pi) =1 $. So, $a$ and $d$ form a blocking pair. 
	
	This implies, however, that the swap dynamics from $\pi$ will cycle and thus do not converge.
	This result is regardless of which swap is selected throughout the swap dynamics, since it can be confirmed that each blocking pair identified above was the unique blocking pair in its respective assignment.
\end{proof}

\vspace{1em}
\noindent \textbf{Proof of Observation \ref{obs: cond1}}
\begin{proof}
	Let $i+1$ be the agent assigned to $v_{i+1}$ under $\pi$. If property (P1) holds, then
	\begin{align*}
		&u_{i+1}(\pi)=\mathbb{I}[\pi^{-1}(v_{i+1})\rightarrow \pi^{-1}(v_{i+2})]+\\
		&+\mathbb{I}[\pi^{-1}(v_{i+1})\rightarrow \pi^{-1}(v_{i})] \\ 
		& \geq \mathbb{I}[\pi^{-1}(v_{i+1})\rightarrow \pi^{-1}(v_{i-1})]+\mathbb{I}[\pi^{-1}(v_{i+1})\rightarrow \pi^{-1}(v_{i})] \\
		&= u_{i+1}(\pi^{i\leftrightarrow i+1}). 
	\end{align*}
	Using a similar argument, we can show that if (P2) holds, then $\pi^{-1}(v_i)$ does not envy $\pi^{-1}(v_{i+1})$. Thus, if either (P1) or (P2) holds, then $\pi^{-1}(v_i)$ and $\pi^{-1}(v_{i+1})$ cannot form a blocking pair.
\end{proof}

\vspace{1em}
\noindent \textbf{Proof of Observation \ref{obs: cond2}}
\begin{proof} Let $i$ be the agent assigned to $v_{i}$ under $\pi$. Then,
	\begin{align*}
		&u_{i}(\pi^{i\leftrightarrow i+1}) =\mathbb{I}[\pi^{-1}(v_{i})\rightarrow \pi^{-1}(v_{i+1})]+ \\
		&+\mathbb{I}[\pi^{-1}(v_{i})\rightarrow \pi^{-1}(v_{i+2})] \\ 
		& \leq \mathbb{I}[\pi^{-1}(v_{i})\rightarrow \pi^{-1}(v_{i+1})] +1 \\
		&= \mathbb{I}[\pi^{-1}(v_{i})\rightarrow \pi^{-1}(v_{i+1})] +  \mathbb{I}[\pi^{-1}(v_{i})\rightarrow \pi^{-1}(v_{i-1})] \\
		&=u_{i}(\pi) .
	\end{align*}
	Hence, $\pi^{-1}(v_i)$ does not envy $\pi^{-1}(v_{i+1})$. A similar argument can be used to prove $\pi^{-1}(v_{i})$ does not envy $\pi^{-1}(v_{i-1})$ if $\pi^{-1}(v_{i})$ approves $\pi^{-1}(v_{i+1})$.
\end{proof}

\end{document}